\documentclass[letterpaper, 10 pt, conference]{ieeeconf}
\IEEEoverridecommandlockouts

\usepackage{amsmath,amssymb,amsfonts}
\usepackage{amsthm}
\usepackage{stmaryrd}
\usepackage{multirow}
\usepackage{mathtools}
\usepackage{acronym}
\usepackage{comment}
\usepackage{todonotes}
\usepackage{bm}
\usepackage{url}
\usepackage{tikz}
\usetikzlibrary{arrows,automata}
\usepackage{subcaption}
\usetikzlibrary{positioning, shapes.geometric}

\newif\ifuseboldmathops
\newif\ifuseittextabbrevs
\useboldmathopstrue   

\ifuseittextabbrevs

	\newcommand{\ie}{{\it i.e.}}

\else

	\newcommand{\ie}{i.e.}

\fi

\ifuseboldmathops
	\newcommand{\reals}{\mathbf{R}}

\else
	\newcommand{\reals}{\mathbb{R}}

\fi

\ifuseboldmathops

\else

\fi

\ifuseboldmathops
	\newcommand{\Expect}{\mathop{\bf E{}}\nolimits}
	
\else
	\newcommand{\Expect}{\mathop{\mathbb{E}{}}\nolimits}
	
\fi

\ifuseboldmathops

\else

\fi

\newcommand{\argmax}{\mathop{\mathrm{argmax}}}

\newcommand{\Eventually}{\Diamond \, }

\newcommand{\cutting}{\begin{tikzpicture}
	\filldraw[fill=black,scale=0.175]
	(0.5, 0.86603)-- (0,0) -- (0.5, -0.86603);
	\filldraw[fill=white, draw=black,scale=0.175]
	 (0.5, 0.86603)--(1,0) -- (0.5, -0.86603);
	\end{tikzpicture}
}
\newcommand{\distEventually}{\cutting }

\newcommand{\until}{\mbox{$\, {\sf U}$}}

\newcommand{\sink}{\mathsf{sink} }

\newcommand{\dist}[1]{\mbox{Dist}(#1)}

\newcommand{\calL}{\mathcal{L}}
\newcommand{\calA}{\mathcal{A}}
\newcommand{\calS}{\mathcal{S}}
\newcommand{\calAP}{\mathcal{AP}}

\newcommand{\calG}{\mathcal{G}}
\newcommand{\calV}{\mathcal{V}}

\newcommand{\init}{\mathsf{Init}}

\renewcommand{\vec}[1]{\mathbf{#1}}
\newcommand{\nat}{\mathbb{N}}
\newcommand{\indicator}{\mathbf{1}}

\DeclareMathOperator{\labelling}{\mathsf{L}}

\DeclareMathOperator{\dta}{\calA}
\DeclareMathOperator{\inv}{\mathcal{I}}

\DeclareMathOperator{\sta}{\calS}

\DeclareMathOperator{\prodmdp}{\mathcal{M}}
\DeclareMathOperator{\prodf}{\mathcal{F}}
\DeclareMathOperator{\prodtrans}{\bar{\Delta}}

\DeclareMathOperator{\clock}{x}
\DeclareMathOperator{\clocks}{X}
\newcommand{\mitldminus}{\mbox{MITLD}$^{-}$ }
\newcommand{\constraints}{\mathcal{C}}
\newcommand{\Occ}{\mathsf{Occ}}
\newcommand{\run}{\rho}

\theoremstyle{definition}
\newtheorem{definition}{Definition}
\newtheorem{example}{Example}
\newtheorem{problem}{Problem}
\newtheorem{lemma}{Lemma}
\newtheorem{assumption}{Assumption}

\acrodef{mdp}[MDP]{Markov Decision Process}
\acrodef{ltl}[LTL]{Linear Temporal Logic}
\acrodef{dfa}[DFA]{Deterministic Finite-state Automaton}
\acrodef{dta}[DTA]{Deterministic Timed Automaton}
\acrodef{sta}[STA]{Stochastic Timed Automaton}
\acrodef{cilp}[CILP]{countably-infinite linear programs}
\acrodef{mitl}[MITL]{Metric Interval Temporal Logic}
\acrodef{mtl}[MTL]{Metric Temporal Logic}
\acrodef{mitld}[MITLD]{Metric Interval Temporal Logic with Probabilistic Distributions}
\acrodef{stl}[STL]{Signal Temporal Logic}

\begin{document}

\title{\LARGE \bf Policy Synthesis for Metric Interval Temporal Logic with Probabilistic Distributions}

\author{Lening Li and Jie Fu
	\thanks{L. Li and J. Fu are with Robotics Engineering Department, Worcester Polytechnic Institute, Worcester, MA 01609, USA,
	{\tt\small lli4, jfu2@wpi.edu.}}%
}

\maketitle

\begin{abstract}
    Metric Temporal Logic can express temporally evolving properties with time-critical constraints or time-triggered constraints for real-time systems. This paper extends the Metric Interval Temporal Logic with a distribution eventuality operator to express time-sensitive missions for a system interacting with a dynamic, probabilistic environment. This formalism enables us to describe the probabilistic occurrences of random external events as part of the task specification and event-triggered temporal constraints for the intended system's behavior. The main contributions of this paper are two folds: First, we propose a procedure to translate a specification into a stochastic timed automaton. Second, we develop an approximate-optimal probabilistic planning problem for synthesizing the control policy that maximizes the probability for the planning agent to achieve the task, provided that the external events satisfy the specification. The planning algorithm employs a truncation in the clocks for the timed automaton to reduce the planning in a countably infinite state space to a finite state space with a bounded error guarantee. We illustrate the method with a robot motion planning example.
\end{abstract}

\section{Introduction}
\label{sec:introduction}
This paper investigates a probabilistic planning problem given a time-sensitive task expressed in Metric Interval Temporal Logic with distributions. The task specifies the assumption about the probabilistic external events and the desired agent's behavior with timing constraints related to the occurrence of the external events. Such problem is widely encountered in robotics \cite{li2019approximate, ding2014optimal}, and other cyber-physical systems \cite{lee2008cyber}. 

Temporal Logic is a formal language to describe desired system properties, such as safety, reachability, obligation, stability, and liveness \cite{manna2012temporal}. Recently, \ac{mtl} \cite{koymans1990specifying}, \ac{mitl} \cite{alur1996benefits}(a fragment of \ac{mtl}), and \ac{stl} \cite{maler2004monitoring} have received extensive attention in the control community, due to their capability of not only expressing the relative temporal ordering of events as \ac{ltl} \cite{pnueli1977temporal} but also defining the duration and timing constraints between these events and robustness of controlled systems. Given high-level specifications, the control synthesis for dynamic systems have been developed for linear systems \cite{lindemann2017robust, yang2020continuous} and linear systems subject to nondeterministic inputs from the environment \cite{raman2015reactive,farahani2015robust}, multi-agent systems \cite{liu2017communication, nikou2016cooperative}, Markov Decision Processes (MDPs) \cite{kapoor2020model}, and Mixed Logical Dynamical systems \cite{saha2016milp}. Mixed integer linear programming has been introduced for robust planning and control in discrete-time, linear systems \cite{raman2015reactive,farahani2015robust} and Mixed Logical Dynamical Systems \cite{saha2016milp}, where the \ac{stl} formula was translated to a set of constraints in the optimization problem. Recent work \cite{yang2020continuous} utilized the Control Barrier Function to plan trajectories that satisfy the \ac{stl} specifications. In particular, they utilized the lower bound of a Control Barrier Function to make the constraint on the \emph{always} temporal operator to be time-invariant. In \cite{saha2016milp}, authors considered specifications expressed in \ac{mtl} for a Mixed Logical Dynamical System. The synthesis of provably correct timed systems has been investigated in the formal methods community using automata-theoretic approaches. The authors in \cite{d2002timed} investigated a control synthesis problem for a timed system, and the specification was captured by an external timed automaton. They modeled the dynamic interaction between the timed system and the environment as a timed game, where the clocks in the timed automaton were assumed to be upper bounded. The authors \cite{xu2019controller} studied multi-agent control with intermittent communication. They solved the problem by developing a leader-follower formulation while the leader must satisfy a \ac{mitl} formula.

Existing works assume that temporal logic specifications only describe the intended system's behavior where the environment is either static \cite{xu2019controller,yang2020continuous} or nondeterministic \cite{d2002timed,raman2015reactive,farahani2015robust}. However, in practical applications, an event in the environment may follow a probability distribution. To illustrate, consider that an autonomous robot is working in a remote area and must catch a bus when a bus arrives. The robot is equipped with the knowledge that the bus's arrival (event) follows a probability distribution. An autonomous car, stopped at a red light, may have a probability distribution of how long to wait for the light to turn green. A patrolling robot may need to interdict illegal fishing, have prior knowledge about when the illegal activities may occur at different sites, and must respond to each detected activity within a bounded time duration. How do we incorporate this knowledge into the specification and how to synthesize a policy that maximizes the probability of satisfying such a time-sensitive task?

Our approach follows from assume-guarantee reasoning:
We adopt the \ac{mitl} to describe the intended system's behavior (guarantee). To capture the assumption of the environment, we consider a new class of distributional temporal logic, proposed recently in \cite{kovtunova2018cutting}. In this logic language, an operator, called ``distribution eventuality'' $\distEventually_\mu \varphi$, is introduced to specify that the first time when $\varphi$ evaluated true follows the distribution $\mu \colon [0, \infty) \to [0,1]$. We augment \ac{mitl} with the distribution eventuality operators, called \ac{mitld}, to specify the intended behavior of the system and the model of its environment jointly. To this end, we propose a procedure to first translate the \ac{mitld} formula to a stochastic timed automaton, which is a finite-state automaton augmented with a set of clocks to keep track of time-dependent constraints. Given the interaction between the agent and its environment modeled as a two-player game, we introduce a product operation between the game and the stochastic timed automaton in order to restrict the external events in the environment to satisfy the pre-defined distributions. This product operation reduces the game into a countably infinite \ac{mdp} with a reachability objective. We introduce a truncation in clocks to reduce the product \ac{mdp} with countably infinite states into one with finite states and ensure the solution found in the latter is at least $\epsilon$-optimal for the original product \ac{mdp}, where the $\epsilon$ is dependent on the truncation points. We use a running example to demonstrate our proposed method and provide discussions for future work. 

\section{Preliminaries}
\label{sec:preliminaries}
\paragraph*{\textbf{Notation}}
We use the notation $\Sigma$ for a finite set of symbols, also known as the \emph{alphabet}. An infinite sequence of symbols $\sigma = \sigma_0 \sigma_1 \ldots \ldots$ with $\sigma_i \in \Sigma$ for all $i \ge 0$, is called an \emph{infinite word}, and $\Sigma^\omega$ is the set of all infinite words that are obtained by concatenating the elements in alphabet $\Sigma$ infinitely many times. Given a random variable $X$, we let $\dist{X}$ be the set of all possible probability distributions over $X$. 

We consider a probabilistic planning problem for a controllable agent, referred to as a robot. The robot aims to satisfy a task specification in an uncertain, stochastic environment. Specifically, we consider the specification of the robot's intended behavior and the environment's behavior is expressed in a class of Metric Interval Temporal Logic with Probability Distributions (MITLD). 

\subsection{Metric Interval Temporal Logic with Probability Distributions}
The \emph{\ac{mitld}} extends \ac{mitl} \cite{alur1996benefits} with a new operator $\distEventually_{\mu} \varphi$, called \emph{distribution eventuality} \cite{kovtunova2018cutting} that expresses that the time until $\varphi$ evaluated true has a probability distribution $\mu$. Formally, \ac{mitld} formulas are built on \ac{mitl} formulas and defined inductively as follows.
\[
	\varphi \coloneqq \top \mid p \mid \neg \varphi \mid \varphi \lor \varphi \mid \varphi \until_{I} \varphi \mid \varphi \until \varphi \mid \distEventually_{\mu} \varphi,
\]
where $p \in \calAP$, $\mu$ is a probability distribution, $\top$ is unconditional true, and $I$ is a \emph{nonsingular} time interval with integer end-points. We also define temporal operator $\Eventually \varphi = \top \until \varphi$ (eventually, $\varphi$ evaluates true in the future) and $\Eventually_{I} \varphi = \top \until_{I} \varphi$ (eventually, $\varphi$ evaluates true within interval $I$ from now). The semantics of \ac{mitl} and the distribution eventuality can be found in \cite{alur1996benefits,kovtunova2018cutting} and omitted.

Note that we do not allow the distribution eventuality operator $\distEventually_{\mu}$ to appear in the scope of any negation since $\neg \distEventually_{\mu} \varphi$ simply says that the time until $\varphi$ evaluated true has a distribution different from $\mu$, which is an uninformative statement. When the context is clear, we may omit the subscript $\mu$ from these formulas. 

This distribution eventuality is introduced by \cite{kovtunova2018cutting} for Linear Temporal Logic formulas. We extend it to use with Metric Interval Temporal Logic and also restrict it to a subclass of \ac{mitld} formulas, termed as \mitldminus, that satisfy the following conditions: 1) The formula has countably many \ac{mitl} interpretations with discrete-time semantics. 2) There exists a subset $U\subset \calAP$ of atomic propositions called ``external events'' such that the distribution eventuality operator can only occur in the front of a proposition in $U$. 3) If we replace the distribution eventuality operator with eventuality operator, the resulting formula is an \ac{mitl} formula that can be equivalently expressed by a \emph{deterministic} timed automaton, introduced later.

We provide some preliminary about timed words and timed automata.

A \emph{(infinite) timed word} \cite{alur1994theory} over $\Sigma$ is a pair $w = (\tau, \sigma)$, where $\tau = \tau_0\tau_1 \ldots $ is an infinite \emph{timed sequence} and $\sigma = \sigma_0 \sigma_1 \ldots \in \Sigma^\omega$ is an infinite word. The infinite timed sequence $\tau = \tau_0\tau_1 \ldots $ satisfies:
\begin{IEEEitemize}
	\item \emph{Initialization}: $\tau_0 = 0$;
	\item \emph{Monotonicity}: $\tau$ increases strictly monotonically; \ie, $\tau_i < \tau_{i+1}$, for all $i \ge 0$;
	\item \emph{Progress}: For every $t \in \reals$, there exists some $i\ge 1$, such that $\tau_i>t$.
\end{IEEEitemize}
The conditions ensure that there are finitely many symbols (events) in a bounded time interval, known as \emph{non-Zenoness}. We also write $w = (\tau, \sigma) = (\tau_0, \sigma_0) (\tau_1, \sigma_1) \ldots $.

A fragment of \ac{mitl} that can be translated into an equivalent \ac{dta} \cite{alur1996benefits}. To define timed automata, a finite number of \emph{clock} and \emph{clock} constraints are needed: Let $\clocks = \{\clock_1, \clock_2, \ldots, \clock_M\}$ be a finite indexed set of clocks. For each clock $\clock_i \in \clocks$, we let $\calV_i$ be the range of that clock. A \emph{clock vector} \cite{bertrand2014stochastic} $v \in \calV_i$ is the vector where the $i$-th entry $v[i]$ is the value of clock $\clock_i$, for $i \in \{1, 2, \ldots, M\}$. We use $\vec{0}$ for the clock vector $v$ where $v[i]=0$ for all $i \in \{1, 2, \ldots, M\}$ and $\calV = \calV_1 \times \ldots \times \calV_M$ for the set of all possible clock vectors. For a clock vector $v \in \calV$ and $t \ge 0$, $v + t$ defines a clock vector $(v + t)[i]= v[i] + t$ for $i \in \{1, 2, \ldots, M\}$. Moreover, if $Z \subseteq \clocks$ is a subset of clocks and $v$ is a clock vector, $v_{[Z \leftarrow 0]}$ denotes the clock vector after resetting the clocks in $Z$, such that $v_{[Z \leftarrow 0]}[i] = v[i]$ for all $\clock_i \in \clocks \setminus Z$ and $v_{[Z \leftarrow 0]}[i] = 0$ for all $\clock_i \in Z$.

\begin{definition}[Clock Constraints \cite{alur1994theory}]
	For a set $\clocks$ of clocks, the set $\constraints(\clocks)$ of clock constraints is defined inductively as follows: 
	\[
		c \coloneqq \perp \mid \clock \bowtie k \mid \clock-\clock' \bowtie k \mid c\land c \mid c\lor c,
	\]
	where $\clock, \clock' \in \clocks$ are clocks, $k \in \nat$ is a non-negative integer, and $\bowtie \in \{=,\ne , <, >, \ge, \le \}$ is a comparison operator.
\end{definition}
Given a clock constraint $c$ over $\clocks$ and a clock vector $v$, we write $v \models c$ whenever $v$ satisfies the constraints expressed by $c$.

Let us recall the definition of the deterministic timed automaton that accepts an \ac{mitl} formula $\varphi^d$. The translation from a \ac{mitl} formula to a \ac{dta} is introduced in \cite{alur1996benefits,maler2006mitl}.
\begin{definition}[Deterministic Timed Automaton]
	A \emph{Deterministic Timed Automaton} associated with an \ac{mitl} formula $\varphi^d$ is a tuple $\dta = \langle L, \Sigma, \init, E, \inv, F \rangle$ with the following components:
	\begin{IEEEitemize}
		\item $L$ is a finite set of locations.	
		\item $\Sigma=2^\calAP$ is a set of symbols.
		\item $\init$ is an initial location.
		\item $E \colon L \times \Sigma \times \constraints(\clocks) \times L\times 2^{\clocks}$ is a finite set of labeled edges. 
		\item $\inv \colon L \to \constraints(\clocks)$ assigns an invariant (clock constraint) to each location.
		\item $F \subseteq L$ is the set of accepting locations. 
	\end{IEEEitemize}
\end{definition}
A \emph{state} in the timed automaton is $(\ell, v)$ consisting of a pair of location and clock vector. The initial state in the timed automaton is $(\init, \vec{0})$. The transition $(\ell, v) \xrightarrow[\tau]{\sigma} (\ell', v')$ occurs only if there exists a symbol $\sigma \in \Sigma$, a constraint $c \in \constraints(\clocks)$, and $\clocks'\subseteq \clocks$ such that $e = (\ell, \sigma, c, \ell',\clocks')$ is defined, and for every $0 \le t < \tau$, $v + t \models \inv(\ell)$, and $v + \tau \models c$. Let $v' = (v+ \tau)_{[\clocks' \leftarrow 0]}$ and $v' \models \inv(\ell')$. Intuitively, the timed automaton reads the input and also checks if the clock constraint of the edge is satisfied given the clock vector. If yes, then the transition occurs and a subset of clocks will be reset to zeros. After transition, the new clock vector should satisfy the invariant condition at the new location. 

By the definition of transition, a \emph{run} $\run$ in $\dta$ on a timed word $w = (\tau_0, \sigma_0) (\tau_1, \sigma_1) \ldots $ is an infinite sequence of transitions $\run = (\init, \vec{0})\xrightarrow[\tau_0]{ \sigma_0}(\ell_1, v_1)\xrightarrow[\tau_1-\tau_0]{ \sigma_1} \ldots (\ell_i,v_i)\xrightarrow[\tau_{i+1}-\tau_i]{ \sigma_{i+1}}(\ell_{i+1}, v_{i+1})\ldots $. The word $w = (\tau, \sigma)$ is accepted by the \ac{dta} if and only if its run reaches an accepting location in $F$; that is, $\Occ(\run) \cap (F \times \calV )\ne \emptyset$, where $\Occ(\run)$ is the set of states in $L \times \calV$ occurring in $\run$. The set of timed words accepted by $\dta$ is called \emph{its language}, denoted as $\calL(\dta)$.

\section{Problem formulation}
Recall that $U \subseteq \calAP$ is the set of atomic propositions such that for each $\alpha \in U$, $\distEventually_{\mu_{\alpha}}\alpha$ is a subformula of the task specification $\varphi$. 
We assume that the agent's action cannot change the values of propositions in $U$. The environment player controls the values of atomic propositions in $U$, and for any $\alpha \in U$, it only evaluates true only once. During the interaction, the agent will select an action, and the environment will select a subset of external events, and the next state is reached probabilistically according to a transition function.

This interaction between agent and its dynamic world can be captured by a two-player turn-based probabilistic transition system. The factored state means that a state $s\in S$ is a vector $s= [s_r,s_e, Y]^\intercal$ of robot state $s_r \in S_r$, environment state $s_e \in S_e$, and a set $Y\subseteq U$ that is a set of external events have \emph{not} occurred yet. The set $S = S_r \times S_e \times 2^U$ is a finite set of possible states reachable given the interaction.

\begin{definition}The interaction between the robot and its environment is modeled as a two-player probabilistic transition system that is a tuple
	\[
		\calG = \langle S \cup S \times A, A \cup 2^U, P, \labelling, s_0, \calAP \rangle
	\]
	with the following components:
	\begin{IEEEitemize}
		\item $S \cup S \times A$ is a set of states partitioned into $S$ and $S \times A$. At each state in $S$, the robot takes an action, and at each state in $S \times A$, the environment takes an action.
		
		\item $A$ is the set of robot actions. $2^U$ is a set of environment actions.
		
		\item At each state in $S$, $P \colon S \times A \to S \times A$ is the deterministic transition function; that is, at a state $s \in S$, given the robot action $a \in A$, the transition to the next state $(s, a)$ is with probability $P(s, a, (s, a)) = 1$. At each state in $S \times A$, $P \colon S \times A \times 2^U \to \dist{S}$ is the probabilistic transition function; that is, at a state $(s, a) \in S \times A$, let $s=(s_r, s_e, Y)$, the environment player selects an \emph{enabled} action $e \in 2^Y$, and $P((s, a), e, s') $ is the probability of reaching state $s'$ given the robot action $a$ from state $s$ and the environment action $e$. In addition, the state $s' = (s_r', s_e', Y')$ in the support of $P((s,a),e)$ must satisfy $Y'= Y\setminus e$.
		
		In other words, the environment can only select a subset of events that have not occurred yet. Given the choice of the external events, the state will be updated to record the occurred events. 
		\item $\labelling \colon S \to 2^\calAP$ is the labeling function such that $\labelling(s)$ is a set of atomic propositions evaluated true at that state. For any $s$ that is reachable with the environment action $e$ with a nonzero probability; that is, $P((s, a), e, s') > 0$, the labeling must satisfies $\labelling(s')\cap U=e$.
		
		\item $s_0$ is an initial state.
	\end{IEEEitemize}
\end{definition}
In this game model, the two players can make nondeterministic choices of actions. The labeling of states captures which events in $U$ have occurred in the most recent transition. 

We informally state the problem to be solved.
\begin{problem}
Given the game model $\calG$ and a task formula $\varphi = \varphi_A\land \varphi_G$ where $\varphi_G$ is the intended robot's behavior in \ac{mitl} and $\varphi_A$ is the assumption about the environment dynamics expressed in a \mitldminus formula, how to synthesize a control policy for the robot that maximizes the probability of satisfying the task formula?
\label{problem:informal_statement}
\end{problem}

\section{Main Results}
\label{sec:main}
In this section, we present our solution to the planning problem. 

\subsection{Translating \mitldminus formulas to Stochastic Timed Automata}
We first construct a computational model that represents the given task formula $\varphi$ in \mitldminus. To do so, we first substitute the distribution eventuality operators $\distEventually$ with the bounded eventuality operators $\Eventually_{\ge 0}$. This substitution will allow us to obtain an \ac{mitl} formula $\varphi^d$ and its corresponding \ac{dta} $\dta$ using the construction from \cite{alur1996benefits,maler2006mitl}. After the substitution, we modify transitions to construct a corresponding \ac{sta} for the original formula with distribution eventuality operators. The following assumption is made about the external events.

\begin{assumption}
	The set of random variables $\{\alpha \mid \alpha \in U\}$ are mutually independent.
\end{assumption}

Recall that the time until the formula $\alpha$ is first observed has a distribution $\mu_\alpha \colon [0, 
\infty) \to [0,1]$. Next, we introduce a stochastic timed automaton to represents the \mitldminus formula $\varphi$. 
\begin{definition}[Stochastic Timed Automaton]
	Given an \ac{mitl} formula $\varphi^d$ with the corresponding \ac{dta} $\dta = \langle L, \Sigma, \init, E, \inv, F \rangle$, the \ac{sta} that represents the distributions of \ac{mitl} formulas expressed by the \mitldminus formula $\varphi$ is a tuple
	$\sta = \langle L, \Sigma, \init, \Delta, \inv, F \rangle$
	with the following components:
	\begin{IEEEitemize}
		\item $L, \Sigma,\init, \inv, F$ are the same as in the \ac{dta} $\dta$ for $\varphi^d$.
		\item $\Delta \colon L \times \calV \times 2^U \times (\Sigma \times \constraints(\clocks)) \times (L \times \calV \times 2^{\clocks} \times 2^U) \to [0,1]$ is a \emph{probabilistic} transition function, constructed as follows. 
	\end{IEEEitemize}
\end{definition}

A state in the stochastic time automaton $\sta$ is $(\ell,v, Y)$ which includes a location, a clock vector, and an additional set of external events. The transition $(\ell, v, Y) \xrightarrow[\tau]{\sigma, p} (\ell', v', Y')$ occurs with a probability $p > 0$ only if $e = (\ell, \sigma, c, \ell', \clocks') \in E$ is defined in the \ac{dta}, $Y' = Y \setminus \sigma$, and for every $0 \le t < \tau$, $v + t \models \inv(\ell)$, $v + \tau \models c$, $v' = (v+ \tau)_{[\clocks' \leftarrow 0]}$, and $v' \models \inv(\ell')$. In addition, the clocks for occurred events, \ie, propositions in $U\setminus Y'$, will be reset and stopped. To define probability $p$ of this transition, we distinguish four cases:

\noindent Case 1: if $Y = \emptyset$, then 
$p = 1$. This is the case when all external events has occurred.

\noindent Case 2: if $Y \ne \emptyset$ and $\sigma \cap Y \ne \emptyset$ and $Y \setminus \sigma \ne \emptyset$, then
\begin{align*}
	p = \prod_{\alpha\in \sigma \cap Y}\mu_\alpha(v(\clock_\alpha)+\tau)\prod_{\beta \in Y \setminus \sigma} (1-\mu_\beta( v(\clock_\beta)+\tau)). 
\end{align*}
Case 3: if $Y \ne \emptyset$ and $\sigma \cap Y \ne \emptyset$ and $Y \setminus \sigma = \emptyset$, then
\begin{align*}
	p = \prod_{\alpha\in \sigma \cap Y}\mu_{\alpha}(v(\clock_\alpha)+\tau). 
\end{align*}
Case 4: if $Y \ne \emptyset$ and $\sigma \cap Y = \emptyset$ and $Y \setminus \sigma \ne \emptyset$, then
\begin{align*}
	p = \prod_{\beta \in Y \setminus \sigma} (1-\mu_{\beta}( v(\clock_\beta)+\tau)). 
\end{align*}
Note that we have augmented the timed automaton's state with a set $Y$ that records a subset of propositions in $U$ that has not been observed given the history. This is because it is known that the external events only occur once, and the distribution eventuality operator only describes the probability of the first occurrence. However, if $\alpha$ is in $Y$, then $\alpha$ has not been observed yet, and the occurrence of $\alpha$ follows its probability distribution. 

For ease of understanding, we provide an example. 
\begin{example}
	Let's consider a \mitldminus formula $\varphi =(\distEventually_{\mu_{1}} b_1 \land \Eventually (b_1 \land \Eventually_{\le 3} (b_3) ) \lor \distEventually_{\mu_2} b_2\land \Eventually (b_2\land \Eventually_{\le 3} (b_4) )$, where the random variables $b_1, b_2$ occur following distributions $\mu_1$ and $\mu_2$. The atomic propositions $\calAP$ are defined as follows:
	\begin{IEEEitemize}
		\item $b_1$: bus $1$ arrives.
		\item $b_2$: bus $2$ arrives.
		\item $b_3$: robot arrives at bus $1$'s station.
		\item $b_4$: robot arrives at bus $2$'s station.
	\end{IEEEitemize}
	Note that $U = \{b_1, b_2\}$.
	The bus $1$ and bus $2$'s arrivals follow probability distributions $\mu_1$ and $\mu_2$. If either bus arrives, then the robot needs to reach the respective bus station within $3$ steps. Note that in this case, we have a set $\clocks = \{\clock_1, \clock_2, \clock_{3}, \clock_{4}\}$ of clocks, where $\clock_1$ and $\clock_2$ reset when $b_1$ and $b_2$ occur, receptively; the clocks $\clock_{3}$ and $\clock_{4}$ are clocks on $b_3$ and $b_4$ triggered after seeing $b_1$ and $b_2$, respectively.
	
	First, the \ac{dta} for the $\varphi^d = \Eventually_{\ge 0} (b_1 \land \Eventually_{\le 3} (b_3) ) \lor \Eventually_{\ge 0} (b_2\land \Eventually_{\le 3} (b_4) )$ is shown in Fig. \ref{fig:example_dta} that is obtained by simplifying the formula after replacing distribution eventuality operators with bounded eventuality operators $\Eventually_{\ge 0}$.
	\begin{figure}[!htb]
		\centering
		\begin{tikzpicture}[->,>=stealth',shorten >=1pt,auto,node distance=5cm,scale=0.6, semithick, transform shape]
			\tikzstyle{every state}=[fill=black!10!white]
			\node[initial, state] 	(0)				        {$\init$};
			\node[state]            (1) [above right of=0]	{$\ell_1$};
			\node[state]            (2) [below right of=0]  {$\ell_2$};
			\node[state]            (4) [below right of=1]  {$\ell_4$};
			\node[accepting, state]	(3) [right of=4] 		{$\ell_3$};
			\path[->]   
			(0) edge [bend left] node       {$\top, \{b_1\}, \{\clock_1,\clock_3\}$}                            (1)
			(0) edge [bend right] node[left]{$\top, \{b_2\}, \{\clock_2,\clock_4\}$}                            (2)
			(0) edge[loop above] node       {$\top, \neg b_1 \land \neg b_2, \emptyset$}                        (0)
			(0) edge node                   {$\top, \{b_1, b_2\}, \{\clock_1, \clock_2, \clock_3, \clock_4\}$}  (4)
			
			(1) edge[loop above]	node {$\top, \neg b_2 \land \neg b_3, \emptyset$}    (1)
			(1) edge				node[left] {$\top, \{b_2\}, \{\clock_2,\clock_4\}$}  (4)
			(1) edge[bend left]		node {$\clock_{3} \le 3, \{b_3\}, \emptyset$}        (3)
			
			(2) edge[loop below]	node {$\top, \neg b_1 \land \neg b_4, \emptyset$} (2)
			(2) edge[bend right]	node[right] {$\clock_{4} \le 3, \{b_4\}, \emptyset$} (3)
			(2) edge				node {$\top, \{b_1\}, \{\clock_1, \clock_3\}$}       (4)
						
			(4) edge[loop above]    node {$\top, \neg b_3 \land \neg b_4, \emptyset$}    (4)
			(4) edge[]    node {$\clock_{3} \le 3, \{b_3\}, \emptyset$}	                 (3)
			(4) edge[bend right]    node[below] {$\clock_{4} \le 3, \{b_4\}, \emptyset$} (3)
			;
		\end{tikzpicture}
		\caption{Deterministic timed automaton $\dta$ for $\varphi^d$, where transitions going to a sink state are omitted.}
		\label{fig:example_dta}
	\end{figure}
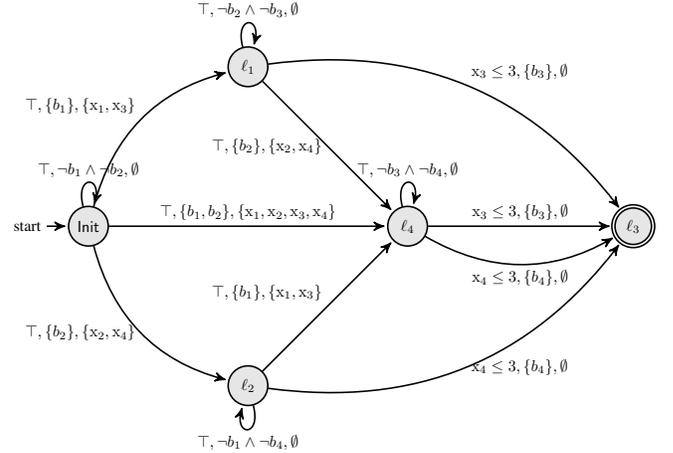
	Given the \ac{dta} in Fig. \ref{fig:example_dta}, we have the corresponding \ac{sta} $\sta$ that augments the locations $L$ with $2^U$. The locations are revised as follows (with the arrow means ``changes into''):
	$\init \to (\init, \{b_1,b_2\})$;
	$\ell_1 \to (\ell_1, \{b_2\})$ (for the event $b_1$ occurred upon reaching $\ell_1$); $\ell_2 \to (\ell_2, \{b_1\})$; $\ell_4\to (\ell_4, \emptyset)$. The accepting state $\ell_3$ will be modified into $(\ell_3, \{b_1\})$, $(\ell_3, \{b_2\})$, and $(\ell_3, \emptyset)$. Also, it is noted that the transitions are probabilistic and non-stationary.

We show a run on the \ac{sta} $\sta$ given a timed word $w = (\tau, \sigma) = (0, \sigma_0 = \emptyset) (1, \sigma_1 = \{b_1\}) (2, \sigma_2 = \emptyset) (3, \sigma_3 = \{b_3\}) \ldots$ as follows:
\begin{align*}
    & (\init, [\overset{\clock_{1}}{0},\overset{\clock_{2}}{0}, \overset{\clock_{3}}{0}, \overset{\clock_{4}}{0}], \{b_1, b_2\}) \\
	& \xrightarrow[0]{\emptyset, p_0} (\init, [\overset{\clock_{1}}{0}, \overset{\clock_{2}}{0}, \overset{\clock_{3}}{0}, \overset{\clock_{4}}{0}], \{b_1,b_2\})     \\ 
	& \xrightarrow[1]{\{b_1\}, p_1} (\ell_1, [\overset{\clock_{1}}{0}, \overset{\clock_{2}}{1}, \overset{\clock_{3}}{0}, \overset{\clock_{4}}{1}], \{b_2\})          \\ 
	& \xrightarrow[1]{\emptyset, p_2} (\ell_1, [\overset{\clock_{1}}{0}, \overset{\clock_{2}}{2}, \overset{\clock_{3}}{1}, \overset{\clock_{4}}{2}], \{b_2\})        \\
	& \xrightarrow[1]{\{b_3\}, p_3} (\ell_3, [\overset{\clock_{1}}{0}, \overset{\clock_{2}}{3}, \overset{\clock_{3}}{2}, \overset{\clock_{4}}{3}], \{b_2\}) \ldots. 
\end{align*}
When the automaton transits from state $(\init, [0, 0, 0, 0], \{b_1, b_2\})$ to state $(\ell_1, [1, 1, 0, 1], \{b_2\})$ after $1$ time unit, the event $b_1$ occurs and thus we reset the clocks $\clock_{1}$ and $\clock_{3}$. After that transition the clock $\clock_1$ remains to be $0$, and other clocks continue to increase without being reset. We list probabilities of transitions about this execution as follows:
\begin{align*}
	& p_0 = (1 - \mu_{b_1}(0)) (1 - \mu_{b_2}(0)), \\
	& p_1 = \mu_{b_1}(1) (1 - \mu_{b_2}(1)),       \\
	& p_2 = (1 - \mu_{b_2}(2)),                    \\
	& p_3 = (1 - \mu_{b_2}(3)).
\end{align*}
\label{example:1}
\end{example}

\subsection{Approximate-optimal planning with \mitldminus formulas} 
Note that $\mu$ is defined over $[0, \infty)$. Due to the difficulty of solving an infinite-state planning problem,
we want to convert the infinite-state model to a finite-state model. This conversion can be made for a subclass of distributions for $\mu_\alpha, \alpha \in U$; that is, all the probability distributions associated with distribution eventuality operators are with a subclass of probability distributions, where given a small constant $\xi$ there always exists a $T$ such that $\sum_{k = T+1}^{\infty} \mu(k) < \xi$. For instance, probability distributions within the exponential family probability distributions \cite{holland1981exponential} satisfy this property. We define a truncation point as follows. 

\begin{definition}[Truncation point]
	For a given formula $\varphi$, we let $tr(\varphi)[i]$ be the truncation point for clock $\clock_i$ given the error bound $\epsilon$ and $tr(\varphi) = [tr(\varphi)[i] \mid i \in \{1, 2, \ldots, M\}]$ be the vector of truncation points for a set $\clocks = \{\clock_1, \clock_2, \ldots, \clock_M\}$ of clocks. The truncated vector for a formula $\varphi$ is computed recursively as follows
	\begin{IEEEitemize}
		\item If $\varphi = \distEventually_{\mu} \alpha$ and clock $\clock_i$ is defined for expressing the distribution eventuality operator, then $tr(\varphi)[i] = T$ such that $\sum_{k=T+1}^{\infty} \mu(k) < \epsilon$; that is, the cumulative probability of first time observing $\alpha$ after truncation point $T$ is smaller than $\epsilon$.
		\item If $\varphi = \varphi_1 \until_{[l,u]} \varphi_2$ and clock $\clock_i$ is defined for expressing the bounded until operator $\until_{[l,u]}$, then $tr(\varphi)[i]= u$, which is the upper bound on the time interval.
		\item If $\varphi = \varphi_1 \land \varphi_2$, then $tr(\varphi) = \max\{tr(\varphi_1), tr(\varphi_2)\}$ where the maximum is element-wise given the vectors $tr(\varphi_1), tr(\varphi_2)$ for common clocks. For a clock that is in $\varphi_1$ but not in $\varphi_2$, $tr(\varphi)$ will include for that clock, and the truncation point is the same as the truncation point for that clock in $tr(\varphi_1)$. 
		\item If $\varphi = \neg \varphi$, then $tr(\varphi) = tr(\varphi)$.
	\end{IEEEitemize}
	\label{def:truncation_point}
\end{definition}
It is noted that the truncation points depend on the error bound $\epsilon$, $tr(\varphi)$ means $tr(\varphi, \epsilon)$. For notation convenience, we just write $tr(\varphi)$ whenever the bound is clear from the context. 
Given the truncated vector $tr(\varphi)$, we denote the set of all possible clock values after the truncation as $\calV^{tr}$, where for any $v \in \calV^{tr}$, we have $v[i] \le tr(\varphi)[i]$ for all $i \in \{1, 2, \ldots, M\}$.

\begin{definition}
	Given a stochastic timed automaton $\sta = \langle L, \Sigma, \init, \Delta, \inv, F \rangle$, a \ac{sta} with truncated clocks $\calV^{tr}$ is a tuple $\sta^{tr} = \langle L, \Sigma, \init, \Delta^{tr}, \inv, F \rangle$
	with the following components:
	\begin{IEEEitemize}
		\item $L$, $\Sigma$, $\init$, $\inv$, $F$ are the same as in the stochastic timed automaton for $\varphi$.
		\item For a transition $(\ell, v, Y) \xrightarrow[\tau]{\sigma, p} (\ell', v', Y')$ defined in $\sta$, if $v ,v'\in \calV^{tr} $, then the same transition occurs in $\sta^{tr}$. If $v\in \calV^{tr}$ but $v'\notin \calV^{tr}$, then $(\ell, v, Y) \xrightarrow[\tau]{\sigma, p} \sink$ where $\sink$ is an absorbing state.
	\end{IEEEitemize}
\end{definition}
Due to the finite clock vectors in $\calV^{tr}$, the states in the \ac{sta} with truncated clocks are finite. 

\begin{lemma}
	For a given formula $\varphi$, there exist $\sta$ and $\sta^{tr}$ with truncated vector $tr(\varphi)$. For a timed word $w$, we have $ P(w \in \calL(\sta_{\varphi}) \land w \not \in \calL(\sta^{tr}_{\varphi})) < \epsilon$.
\end{lemma}
\begin{proof}[Proof (Sketch)]
    We first prove for the base case, where there is at most one clock in any subformula $\varphi$. 
    \begin{IEEEitemize}
        \item Given $\varphi = \distEventually_{\mu} \alpha$ where $\alpha\in U$, a timed word $w = (\tau_0, \sigma_0) (\tau_1, \sigma_1) \ldots (\tau_n, \sigma_n) \ldots \models \varphi$ if there exists an index $n > 0$, such that $\alpha \in \sigma_n$, then we have $P(w \in \calL(\sta_{\varphi}) \land w \not \in \calL(\sta^{tr}_{\varphi})) = P(n > T)$, where $T = tr(\varphi, \epsilon)$. Given the definition of the truncation point $tr(\varphi, \epsilon)$, we have $P(n > T) = \sum_{k=T+1}^{\infty}\mu(k) < \epsilon$.
        \item If $\varphi=p_1\until_{[l,u]}p_2$, then $tr(\varphi)=u$, a word $w$ accepted in $\sta$ will be accepted in $\sta^{tr}$ as the clock is upper bounded and the maximal value is used as the truncation point.
        \item If $\varphi = \neg \varphi$, due to the choice of the truncated vector $tr(\varphi) = tr(\neg \varphi)$, there is no error introduced by the negation operator.
    
        \item If $\varphi = \varphi_1 \land \varphi_2$, assume that we have ${P(w \in \calL(\sta_{\varphi_1}) \land w \not \in \calL(\sta^{tr}_{\varphi_1}))} < \epsilon$ and ${P(w \in \calL(\sta_{\varphi_2}) \land w \not \in \calL(\sta^{tr}_{\varphi_2}))} < \epsilon$. We assume events in $\varphi_1$ and $\varphi_2$ are independent, and we distinguish three cases:
		\begin{IEEEitemize}
		    \item Case 1: If $w \in \calL(\sta^{tr}_{\varphi_1})$ and $w \not \in \calL(\sta^{tr}_{\varphi_2})$, then we have ${P(w \in \calL(\sta_{\varphi}) \land w \not \in \calL(\sta^{tr}_{\varphi}))} < \epsilon \cdot (1 - \epsilon)$
		    \item Case 2: If $w \not \in \calL(\sta^{tr}_{\varphi_1})$ and $w \in \calL(\sta^{tr}_{\varphi_2})$, then we have ${P(w \in \calL(\sta_{\varphi}) \land w \not \in \calL(\sta^{tr}_{\varphi}))} <(1 - \epsilon) \cdot \epsilon$
		    \item Case 3: If $w \not \in \calL(\sta^{tr}_{\varphi_1})$ and $w \not \in \calL(\sta^{tr}_{\varphi_2})$, then we have ${P(w \in \calL(\sta_{\varphi}) \land w \not \in \calL(\sta^{tr}_{\varphi}))} <\epsilon^2$
		\end{IEEEitemize}
		There only one case can be true for every timed word $w$, so we have ${P(w \in \calL(\sta_{\varphi}) \land w \not \in \calL(\sta^{tr}_{\varphi}))} < \max\{\epsilon (1 - \epsilon), \epsilon^2\} < \epsilon$ as $\epsilon \le 1$. 
	
		\item If $\varphi = \varphi_1 \until_{[l, u]} \varphi_2$, we can write the timed word $w$ as $w_1 w_2$ and use the same technique to prove ${P(w_1w_2 \in \calL(\sta_{\varphi}) \land w_1w_2 \not \in \calL(\sta^{tr}_{\varphi}))} < \epsilon$. Due to the space limitation, we omit the proof.
	\end{IEEEitemize}

    The vector of truncation points is computed recursively in Def. \ref{def:truncation_point}, we can conclude that ${P(w \in \calL(\sta_{\varphi}) \land w \not \in \calL(\sta^{tr}_{\varphi}))} < \epsilon$.
\end{proof}

\subsection{Composition: Product Markov Decision Processes}
Given that the specification of the environment informs us that the events in $U$ are not nondeterministic but following given distributions, we now introduce a product operation between the two-player transition system and the stochastic timed automaton with truncated clocks. With this product, we reduce the game to a one-player stochastic game, also known as a \ac{mdp} with a reachability objective.
\begin{definition}[Product \ac{mdp}]
	Given a two-player turn-based probabilistic transition system $\calG = \langle S \cup S \times A, A \cup 2^U, P, \labelling, s_0, \calAP \rangle$ and a \ac{sta} with truncated clocks $\sta^{tr} = \langle L, \Sigma, \init, \Delta^{tr}, \inv, F \rangle$ associated with the \mitldminus formula $\varphi$, a product Markov Decision Process is a tuple
	\[
		\prodmdp = \calG \times \sta^{tr} = \langle Z, A, \prodtrans, z_0, \prodf \rangle
	\]
	with the following components:
	\begin{IEEEitemize}
		\item $Z = S \times (L \times \calV^{tr} \times 2^U \cup\{\sink\})$ is the set of product states. For notation convince, we denote $H\coloneqq (L \times \calV^{tr} \times 2^U \cup\{\sink\})$. 
		\item $A$ is the set of robot actions.
		$\prodtrans: Z\times A \rightarrow \dist{Z}$ is defined as follows. At a state $(s, h)$, the robot selects action $a\in A$, $\prodtrans((s,h),a,(s',h')) = P((s,a),e, s') \cdot p$, where $p$ is the probability of the transition $h\xrightarrow[1]{\labelling(s'),p}h'$.
		 
		\item $z_0 = (s_0, (\ell_0, \vec{0}, U))$ is the initial state that includes the initial state $s_0$ in the two-player turn-based probabilistic transition system, and $(\init, \vec{0}, U) \xrightarrow[0]{\labelling(s_0), 1} (\ell_0, \vec{0}, U)$.
		\item $\prodf =S\times F \times \calV^{tr} \times 2^U$ is the set of final states, where $F$ denotes the set of accepting states in the \ac{sta}. The states in $\prodf$ are absorbing states.
	\end{IEEEitemize}
 \end{definition}

We define a reward function over the product \ac{mdp} as follows. For a product state $z \in Z$, an action $a \in A$, and the next product state $z' \in Z$,
\begin{align*}
	R(z, a, z') = \indicator(z \in Z \setminus \prodf) \indicator(z' \in \prodf).
\end{align*}
This reward function describes that a reward of one is received only if the robot transits from a product state not in $\prodf$ to a product state in $\prodf$, and we have the expected reward by choosing action $a$ at state $z$, \ie, $R(z, a) = \sum_{z'} \prodtrans(z' \mid z, a) R(z, a, z')$.

Given the product \ac{mdp} $\prodmdp$, to maximize the probability of satisfying the \mitldminus formula $\varphi$ starting at an initial product state $z_0$ is equivalent to maximize the total rewards:
$\pi^\ast = \argmax_{\pi} \Expect[\sum_{t=0}^{\infty} R(Z_t, A_t) \mid Z_0 = z_0, \pi]$. Given the product $\prodmdp = \langle Z, A,\prodtrans, z_0, \prodf \rangle$, we adopt the model-based value iteration to solve for the optimal value function and policy.

\section{Case Study}
\label{sec:case_study}
We consider a motion planning problem inspired by the running example \ref{example:1}, a robot catching buses problem. As shown in Fig. \ref{fig:grid_world}, there exists a robot whose goal is to maximize the probability of catching buses within a $4 \times 4$ grid world. A bouncing wall surrounds the grid world, \ie, if the robot hits the wall, it stays at the previous cell. The robot has the action set $\{\text{N, W, E, S}\}$ for moving in four compass directions, and all actions have probabilistic outcomes (see Fig.~\ref{fig:grid_world}). There are two bus stations, namely, region $A$ and region $B$, where bus $1$ and bus $2$ arrive according to two independent probability distributions. In this example, we consider two cases: case $1$, $b_1$ and $b_2$ follows geometric probability distributions $\mu_1$ and $\mu_2$ with parameters $0.8$ and $0.3$; case $2$, $b_1$ and $b_2$ follows geometric distributions $\bar{\mu}_1$ and $\bar{\mu}_2$ with parameters with $0.4$ and $0.7$. For the atomic proposition $\alpha \in U$ following a geometric probability distribution $\mu$ with parameter $p$, we have the probability $p$ that $\alpha$ evaluates true and the probability $1-p$ that $\alpha$ evaluates false every time step until the first time $\alpha$ is observed.
\begin{figure}[!htb]
	\centering
	\includegraphics[width=0.8\linewidth]{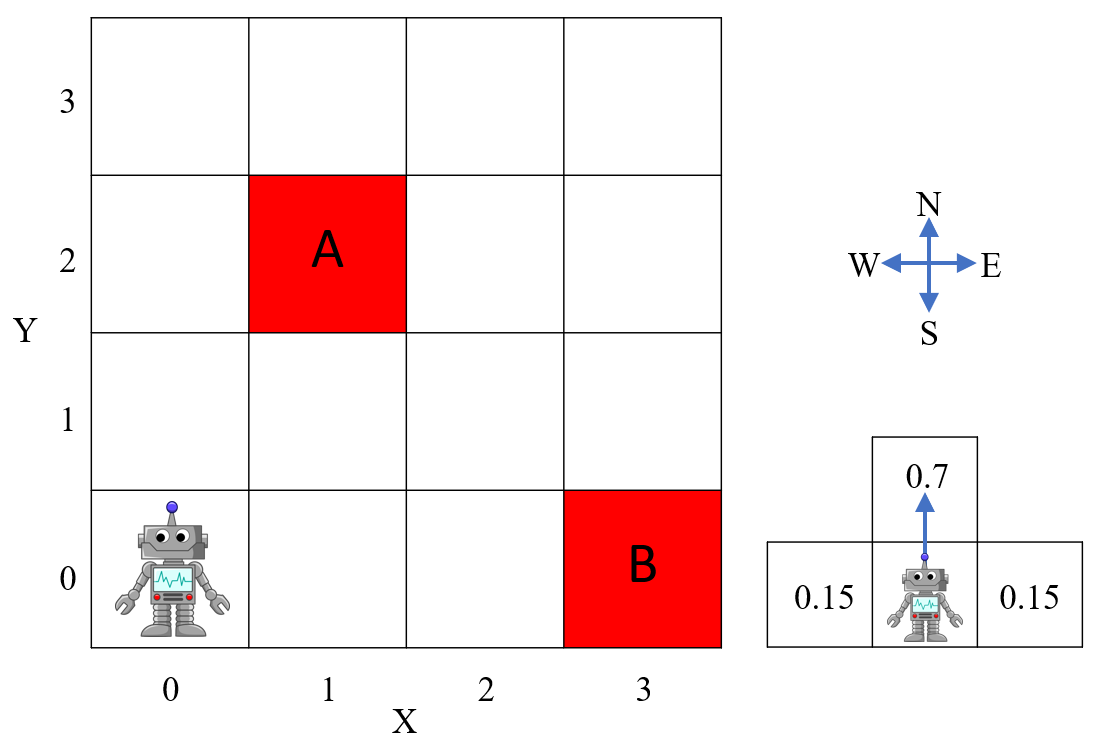}
	\caption{The $4 \times 4$ grid world.}
	\label{fig:grid_world}
\end{figure}

\begin{table}[!htb]
\centering
\caption{Initial state's value versus error bounds $\epsilon_1, \epsilon_2$ under different probability distributions $\mu_1, \mu_2$ and $\bar{\mu}_1, \bar{\mu}_2$.}
\label{tab:peformance}
\resizebox{0.8\linewidth}{!}{%
\begin{tabular}{l|llll}
\hline
& State size & 518400 & 921600 & 1440000 \\ \hline
\multirow{2}{*}{Case 1: $\mu_1, \mu_2$} & $\epsilon_1$ & 0.343 & 0.2401 & 0.16807 \\
                                        & Value      & 0.43264 & 0.60645  & 0.60645 \\ \hline
\multirow{2}{*}{Case 2: $\bar{\mu}_1, \bar{\mu}_2$} & $\epsilon_2$ & 0.216 & 0.1296 & 0.07776 \\
                                        & Value      & 0.51411 & 0.62313 & 0.62313 \\ \hline
\end{tabular}%
}
\end{table}
Given different error bounds associated with corresponding of truncated vectors, we list the initial state' values; that is, the probabilities of robot satisfying the given specification, in Table. \ref{tab:peformance}. The initial state's value increases with decreasing error bounds $\epsilon_1, \epsilon_2$ for both cases. However, the state space in the product \ac{mdp} increases as the error bounds decrease too.

Given the error bound $\epsilon_2 = 0.07776$, we simulate trajectories in the case $2$ using the obtained policy. One trajectory is as follows: 
\begin{align*}
    & (((0, 0), \emptyset, \{b_1, b_2\}), (\init, (0, 0, 0, 0), \{b_1, b_2\})) \\
    & \xrightarrow{\text{E}} ((((0, 0), \emptyset, \{b_1, b_2\}), \text{E}), (\init, (0, 0, 0, 0), \{b_1, b_2\}))) \\
    & \xrightarrow{b_2} ((((1, 0), \{b_2\}, \{b_1\}), (\ell_2, (1, 0, 1, 0), \{b_1\})) 	 \\
    & \xrightarrow{\text{E}} ((((1, 0), \{b_2\}, \{b_1\}), \text{E}), (\ell_2, (1, 0, 1, 0), \{b_1\})) \\
    & \xrightarrow{b_1} (((1, 1), \{b_1\}, \emptyset), (\ell_4, (0, 0, 0, 1), \emptyset)) 	 \\
    & \xrightarrow{\text{N}} ((((1, 1), \{b_1\}, \emptyset), \text{N}), (\ell_4, (0, 0, 0, 1), \emptyset)) \\
    & \xrightarrow{\emptyset} (((1, 2), \emptyset, \emptyset), (\ell_3, (0, 0, 1, 2), \emptyset))
\end{align*}
The state is understood as, for example, $(((0, 0), \emptyset, \{b_1, b_2\})$ means that the robot is at position $(0,0)$, no external events happen in the most recent transition, and $b_1,b_2$ have not been observed. Since in case 2 the atomic proposition $b_2$ has a higher probability of evaluating true every step until the first time $b_2$ is observed, the robot initially takes two consecutive actions E to go to the region $B$, and the bus 2 arrives after it took the first action E. However, due to the stochastic dynamics, the robot reaches state $(1, 1)$ by taking action E at state $(1, 0)$ while bus 1 arrives ($b_1$ evaluates true). Seeing it has more chance to reach bus 1, the robot takes action N to reach region $A$ instead of going to region $B$.

\section{Conclusion}
\label{sec:conclusion}
We investigated how to synthesize an optimal control policy for a new class of formal specifications which extend the Metric Interval Temporal Logic with a distribution eventuality operator $\distEventually$. This new operator allows us to incorporate prior information of uncontrollable events in the interacting environments for decision-making with time-sensitive and event-triggered temporal constraints. We introduce a truncation of clocks to reduce the countably infinite \ac{mdp} into a finite-state model and provide the error bound on the solution. In the future, we will investigate the scalability of our method. With abstraction methods such as region automaton \cite{alur1994theory}, we can potentially reduce the size of the approximated finite-state model. We are also considering incorporating approximate dynamic programming \cite{li2019approximate} for large-scale \ac{mdp}s to handle the issue of scalability.

\bibliographystyle{IEEEtran}			
\bibliography{refs}

\end{document}